\documentclass[a4paper,reqno]{amsart}
\usepackage{amsfonts}
\usepackage[a4paper, margin=2.2cm]{geometry}
\usepackage{tikz}
\usepackage{tkz-graph}
\usepackage{accents}
\usepackage{array,calc} 
\usepackage{stmaryrd} 
\usepackage{pgf}
\usepackage{algorithm}
\usepackage{algpseudocode}

\usepackage{ulem}

\usepackage{subcaption}

\usepackage{tikz}
\usepackage{tikz-cd}
\tikzcdset{column sep/normal=1.1cm}
\usetikzlibrary{calc}
\usepackage{color, mathtools,epsfig, graphicx}
\usetikzlibrary{positioning}
\usetikzlibrary{positioning, shapes, shadows, arrows}

\usepackage{amsmath,calc,graphicx}
\usepackage{url}
\usepackage[hidelinks]{hyperref}
\usepackage{amssymb}
\usepackage{amsthm}
\usepackage{multirow}
\usepackage{float}
\usepackage{enumerate}
\usepackage{enumitem}
\usepackage{amsrefs}
\usepackage{comment}
\usepackage{mathrsfs}
\usepackage{longtable}
\usepackage{todonotes}
\numberwithin{equation}{section}
\numberwithin{figure}{section}

\newtheorem{theorem}{Theorem}[section]
\newtheorem{remark}[theorem]{Remark}
\newtheorem{definition}[theorem]{Definition}
\newtheorem{conjecture}[theorem]{Conjecture}

\newtheoremstyle{break}
  {\topsep}{\topsep}%
  {\itshape}{}%
  {\bfseries}{}%
  {\newline}{}%
\theoremstyle{break}

\usepackage{mathtools}

\makeatletter
\newcommand*\rel@kern[1]{\kern#1\dimexpr\macc@kerna}
\newcommand*\widebar[1]{%
  \begingroup
  \def\mathaccent##1##2{%
    \rel@kern{0.8}%
    \overline{\rel@kern{-0.8}\macc@nucleus\rel@kern{0.2}}%
    \rel@kern{-0.2}%
  }%
  \macc@depth\@ne
  \let\math@bgroup\@empty \let\math@egroup\macc@set@skewchar
  \mathsurround\z@ \frozen@everymath{\mathgroup\macc@group\relax}%
  \macc@set@skewchar\relax
  \let\mathaccentV\macc@nested@a
  \macc@nested@a\relax111{#1}%
  \endgroup
}
\makeatother

\newcommand\Pone{\textrm{P}_{\textrm{I}} }

\newcommand{\orcidauthorA}{0000-0001-7504-4444}
\newcommand{\orcidauthorB}{0000-0002-0461-7580}

\title{Arithmetic dynamics of a discrete Painlev\'e equation}
\date{}
\thanks{This research was supported by the Australian Government through the Office of National Intelligence NISDRG grant \#NI240100145.}
\author[Nalini Joshi]{Nalini Joshi}
\thanks{NJ's ORCID ID is \orcidauthorA.}
\address{School of Mathematics and Statistics F07, The University of Sydney, NSW 2006, Australia}
\author{Pieter Roffelsen}
\thanks{PR's ORCID ID is \orcidauthorB.}
\email{nalini.joshi@sydney.edu.au}
\email{pieter.roffelsen@sydney.edu.au}
\subjclass[2020]{39A13, 33E17,37P05,11G20}

\begin{document}
\begin{abstract}
Motivated by arithmetic dynamics on elliptic curves, we consider bounds on the number of image points of a discrete Painlev\'e equation over finite fields and show a surprising connection to the Hasse bound. Moreover, we show that the orbits of the discrete Painlev\'e equation in a finite field lie on algebraic curves given by polynomials and we give explicit definitions of such polynomials. These results contrast sharply with the theory of discrete Painlev\'e equations over $\mathbb{C}$, where generic solutions are believed not to lie on algebraic curves.
\end{abstract}
\maketitle

\section{Introduction}


In this paper, we consider the orbits of a discrete Painlev\'e equation
\begin{equation}\label{eq:qp1}
\begin{cases}
\displaystyle\overline{x}&=\dfrac{t}{x-s^{-1}y},  \\
\displaystyle\overline{y}&=\dfrac{s\, x}{y},  \\
\displaystyle\overline{t}&=s\,t, 
\end{cases}
\end{equation}
over finite fields $\mathbb F_q$, $q=p^n$, for prime $p$ and integer $n\geq 1$. We show that the orbit lengths satisfy an analog of the Hasse bound and, moreover, that the orbits lie on algebraic curves. These results are surprising because it is known that the solutions of $q\Pone$ are transcendental functions \cite{bernardo,nishioka2009,nishioka2010,nishioka2017} .

Equation \eqref{eq:qp1} is known as $q$-discrete Painlev\'e one ($q\Pone$). In this name, $q$ refers to the iterative parameter, which 
we have renamed $s$, and it should not be confused with the order of the field $\mathbb F_q$. Equation \eqref{eq:qp1} shares remarkable properties with its continuum limit, which is the celebrated first Painlev\'e equation ($\Pone$), obtained as $s\to 1$. One of these is their association with linear systems of equations called Lax pairs, which characterizes $\Pone$ and $q\Pone$ as integrable systems. In this paper, we use the Lax pair given in \cite[\S 2.7]{murata} for $q\Pone$, which remains well-defined over finite fields. Like $\Pone$, $q\Pone$ has an autonomous form, whose solutions lie on elliptic curves \cite{d:10}. Motivated by the longstanding study of elliptic curves in finite fields \cite{silverman}, it is natural to ask what is known about solutions of $q\Pone$ in such fields. 


A large amount of work has been devoted to the study of dynamical systems in finite fields \cite{silvermandynamics}. Two specific directions concerned with integrable systems in finite fields are of interest here. One is the study of autonomous maps in the plane and how bounds on orbits are related to integrability  \cite{JRV06,robertsetal2003}. This corresponds to the case when $s=1$ in Equation \eqref{eq:qp1}. However, in our case there exists $r\geq 1$, a factor of $q-1$, such that $s^r=1$, so that $s\neq 1$ in general, meaning Equation \eqref{eq:qp1} is non-autonomous. The second direction focuses on methods to faithfully define initial value spaces of discrete Painlev\' e equations over $p$-adic fields, finite fields and their relation through `almost good reductions' \cite{kankisigma,kankirims}. These methods enable orbits to be well-defined over finite fields. Our main purpose is to establish distributions of orbit lengths, which has not been a topic of consideration in such earlier studies. 

Our main results give rise to bounds on orbit lengths that are analogous to the Hasse bound. (See Conjecture \ref{conj:numerical}.) Hasse's bound is an estimate of the number of points on an elliptic curve in a finite field. Our estimates differ because they incorporate an additional parameter $r$, which is the multiplicative order of $s$.

\subsection{Main results}
Our main results are summarised in Conjecture \ref{conj:numerical}. These conjectures arise from the study of orbits  of Equation \eqref{eq:qp1} over $\mathbb{F}_q$, with the prime power $q$ ranging from $2$ to $499$. Magma V2.28-23 was used to study over 200M orbits to verify these conjectures; see Section \ref{sec:codedata}. 

These observations lead naturally to the question: do the orbits lie on algebraic curves? We give an affirmative answer in Section \ref{sec:integral}, where the curves are computed explicitly. In Theorem \ref{thm:integrals}, we show that Equation \eqref{eq:qp1} has a rational \textit{integral of motion}, whenever $s$ is a root of unity, regardless of the underlying field. Using a Lax pair, we give an explicit formula to compute them. All orbits lie on algebraic curves defined by fibres of this integral.  In Section \ref{sec:genus}, we study the (geometric) genus of these curves, giving rise to a further conjecture that the genus must be generically one; see Conjecture \ref{conj:genus}.





\begin{definition}
Notation:
\begin{enumerate}[label={\rm (}\roman*{\rm )}, ref=, align=left,leftmargin=*]
    \item Let $\mathbb{F}_{q}^*:=\mathbb{F}_{q}\setminus\{0\}$. Fix $s\in\mathbb{F}_{q}^*$, denote the \textit{multiplicative order} of $s$ by $r$. That is, $r\geq 1$ is the smallest positive integer such that $s^r=1$.
    \item We reserve the symbol $\gamma$ for orbits of Equation \eqref{eq:qp1}. Over a finite field $\mathbb{F}_{q}$, all orbits are periodic, and, denoting by $m\geq 1$ the smallest period of an orbit $\gamma$, we write the orbit as
\begin{equation*}
\gamma=(\gamma_0,\gamma_1,\gamma_2,\ldots,\gamma_{m-1}),  
\end{equation*}
where the states $\gamma_k$, $0\leq k\leq m-1$, satisfy
$\gamma_{k+1}=\overline{\gamma}_k$ for $0\leq k\leq m-1$ and $\gamma_{0}=\overline{\gamma}_{m-1}$. A precise definition of states is given in Definition \ref{eq:defi_states}.
\item We denote the number of states in the orbit by $\#(\gamma)$. The multiplicative order $r$ of $s$ must be a divisor of both $q-1$ and $\#(\gamma)$, and we refer to $\#(\gamma)/r$ as the \emph{reduced orbit length} of $\gamma$.
\end{enumerate}
\end{definition}

\begin{conjecture}\label{conj:numerical}
 Let $\gamma$ denote any orbit of Equation \eqref{eq:qp1} over $\mathbb{F}_{q}$ and suppose $r$ is the multiplicative order of $s\in\mathbb{F}_{q}^*$. Then the following results hold true. 
\begin{enumerate}[label=\emph{\bf 1.2.\Alph*}., ref={1.2.\Alph*}, labelindent=0pt]
\item \label{conj:I} The number of points $\#(\gamma)$ satisfies
   \begin{equation}\label{eq:conj_I_bound}
    \#(\gamma)/r\leq q+2\sqrt{q}+1.
\end{equation} 
\item \label{conj:II}
Let $M_q$ be the positive integer defined by the ceiling
\begin{equation*}
    M_q=\left\lceil\frac{1}{4}\left(\sqrt{q}+\frac{1}{\sqrt{q}}-2\right)\right\rceil ,
\end{equation*}
and define bins
\begin{equation}\label{eq:bindef}
 B_m^{(q)}=\begin{cases}   \left[\frac{q+1-2\sqrt{q}}{m},\frac{q+1+2\sqrt{q}}{m}\right], &1\leq m\leq M_q-1,\\[6pt]
 \left[1,\frac{q+1+2\sqrt{q}}{m}\right], &m=M_q.
 \end{cases}
 \end{equation}
Then $\#(\gamma)/r$ lies in one of the $M_q$ disjoint bins $B_m^{(q)}$, $1\leq m\leq M_q$.
\end{enumerate}
\end{conjecture}
\begin{remark}[Autonomous case]\label{rem:auto}
   When $s=1$, the dynamical system has a bidegree $(2,2)$ integral of motion,
\begin{equation*}
   I_1=y-x+\frac{x}{y}-\frac{t}{x},
\end{equation*}   
whose generic fibres are elliptic curves. The  mapping given in \eqref{eq:qp1} then coincides with  addition on fibres and the statements in Conjecture \ref{conj:numerical} are a direct consequence of the Hasse bound.
\end{remark}
Conjecture \ref{conj:I} is the Hasse upper bound, well known for elliptic curves, and is illustrated in Figure \ref{fig:conjI}; see also Figure \ref{fig:conjIzoom}. From these figures, one can observe  that the bound \eqref{eq:conj_I_bound} is sharp for many of the values of $q$ and $r$ included. 

In Figure \ref{fig:conjI},  for any fixed value of $q$, one can see that the values of reduced orbit lengths only lie in certain bins, demarcated in red in the figure. This is the content of Conjecture \ref{conj:II}, which is a refinement of  \ref{conj:I}. 


\begin{figure}[t!]
    \centering
    \begin{subfigure}[t]{0.5\textwidth}
        \centering
        \includegraphics[width=\textwidth]{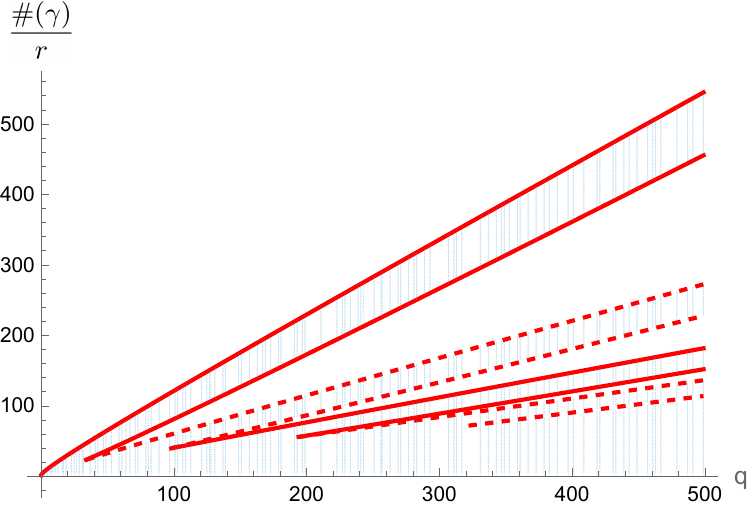}
        \caption{$r=1$}
    \end{subfigure}%
    ~ 
    \begin{subfigure}[t]{0.5\textwidth}
        \centering
        \includegraphics[width=\textwidth]{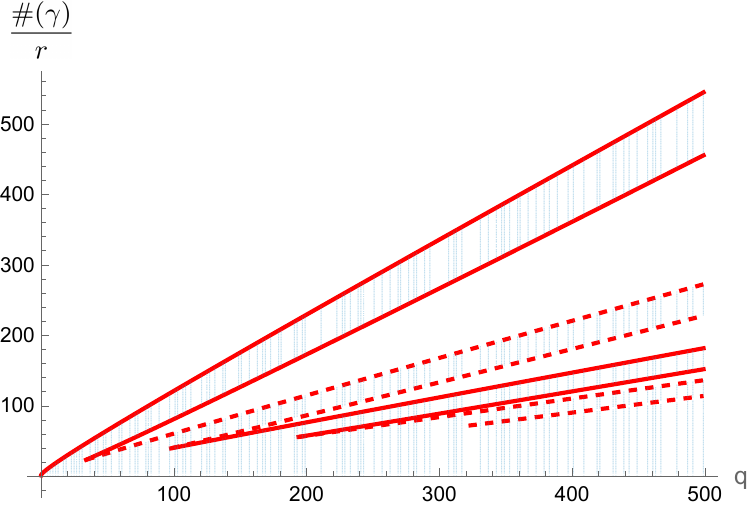}
        \caption{$r=2$}
    \end{subfigure}\\
    \begin{subfigure}[t]{0.5\textwidth}
        \centering
        \includegraphics[width=\textwidth]{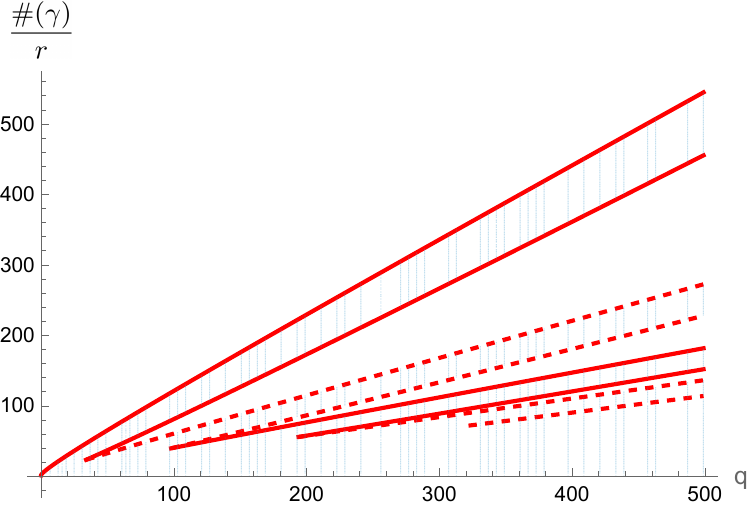}
        \caption{$r=3$}
    \end{subfigure}%
    ~ 
    \begin{subfigure}[t]{0.5\textwidth}
        \centering
        \includegraphics[width=\textwidth]{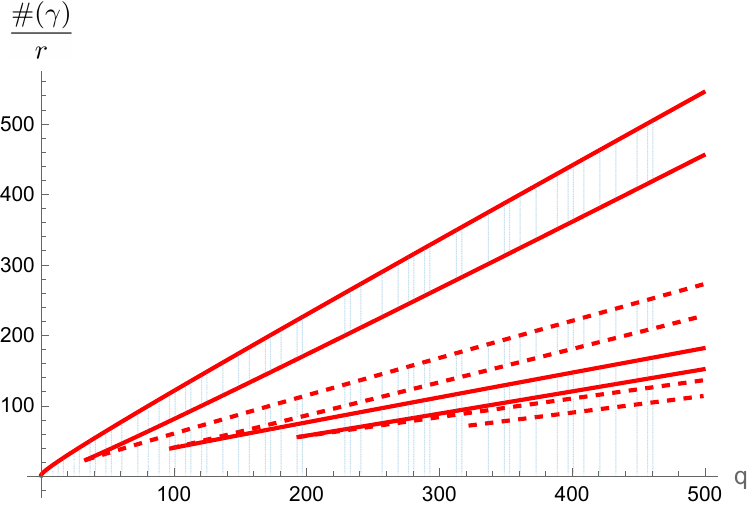}
        \caption{$r=4$}
    \end{subfigure}\\
    \begin{subfigure}[t]{0.5\textwidth}
        \centering
        \includegraphics[width=\textwidth]{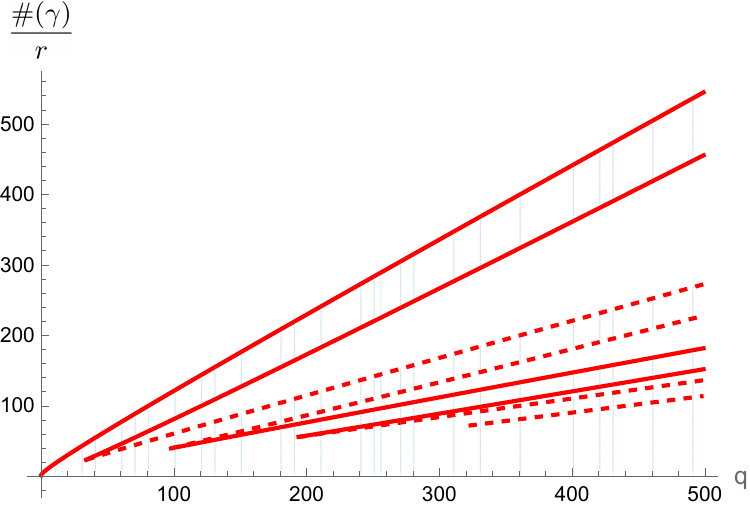}
        \caption{$r=5$}
    \end{subfigure}%
    ~ 
    \begin{subfigure}[t]{0.5\textwidth}
        \centering
        \includegraphics[width=\textwidth]{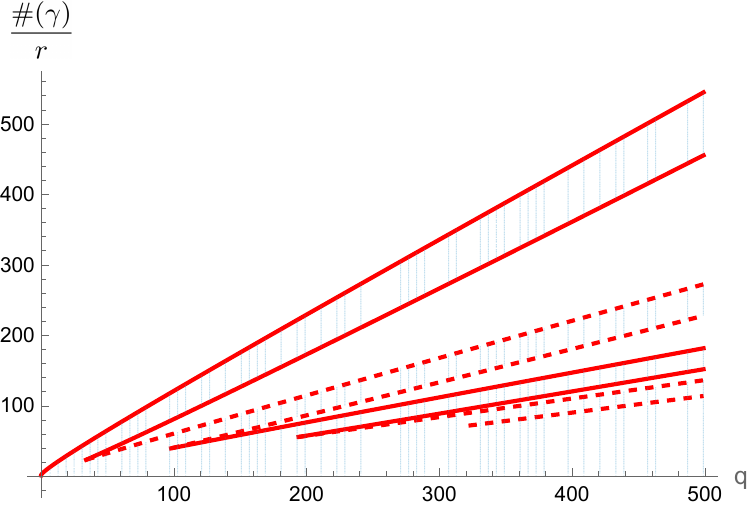}
        \caption{$r=6$}
    \end{subfigure}
    \caption{In each subfigure, for the value of $r$ indicated in the corresponding caption, Hasse's upper bound $q+2\sqrt{q}+1$ is given by the top most solid red curve. Upper and lower bounds for the bins $B_m^{(q)}$, $1\leq m\leq 4$, are displayed alternately in solid and dashed red lines. The different numbers of points on orbits divided by $r$ are shown in blue, for all orbits and all parameter values, for any given prime power $2\leq q\leq 500$ such that $r\,\bigm |\,(q-1)$.}
  \label{fig:conjI}
\end{figure}



\subsubsection{Example} Consider Equation \eqref{eq:qp1} over $\mathbb{F}_3$ with $s=2$. The multiplicative order of $s$ is $r=2$ and the corresponding integral of motion is
\begin{equation*}
    I_2=\frac{2 t^2}{x^2}+\frac{2 t y}{x}+\frac{2 t}{y}+\frac{2 x^2}{y^2}+\frac{2 x^2}{y}+2 x^2+x y+2 x+2 y^2.
\end{equation*}
Multiplying the equation for fibres, $I_2=c$, by $x^2y^2$ yields a pencil of surfaces in the variables $x,y,t$,
\begin{equation}\label{eq:polsurface}
    2 t^2 y^2+2 t x^2 y+2 t x y^3+2 x^4 y^2+2 x^4 y+2 x^4+x^3 y^3+2 x^3 y^2+2 x^2 y^4=c\,x^2y^2.
\end{equation}
Every orbit lies within a particular member, since $I_2$ is constant on orbits. Each of the surfaces is fibred through the projection that maps $(x,y,t)$ to $t$. Its fibres are the algebraic curves defined by equation \eqref{eq:polsurface} for fixed values of $c$ and $t$. The geometric genus of these algebraic curves within $\mathbb{P}^1\times \mathbb{P}^1$ over $\mathbb{F}_3$ is one for all $c\in \mathbb{F}_3$ and $t\in \mathbb{F}_3^*$.


\subsection{Code and data}\label{sec:codedata}
We generated data to support Conjecture \ref{conj:numerical} using Magma and analysed the data in Mathematica. The data and relevant code are available at \cite{orbitdata}. Furthermore, \cite{orbitdata} contains Magma code used to check
Conjecture \ref{conj:genus} as explained in Section \ref{sec:genus}. Organisation of the data is explained in Section \ref{sec:data}, as well as in the README file in \cite{orbitdata}.

\subsection{Outline of the paper}
In Section \ref{sec:numerical}, we provide the computational results to support Conjecture \ref{conj:numerical}. Resolution of singularities over $\mathbb{F}_q$ and algorithms to obtain data on orbits are described there.
In Section \ref{sec:integral}, integrals of motion of Equation \eqref{eq:qp1} are constructed and geometric genera of the algebraic curves they define are studied. The paper concludes with Section \ref{s:conc}, where we summarise our observations and suggest further questions.

\subsection{Acknowledgments}
This research was funded by the Australian Government through the Office of National Intelligence grant \# NI240100145.
The authors would like to thank the Magma Team at the University of Sydney, in particular John Voight and Allan Steel, for insightful conversations and collegial support.

\begin{figure}[t!]
    \centering
    \begin{subfigure}[t]{0.5\textwidth}
        \centering
        \includegraphics[width=\textwidth]{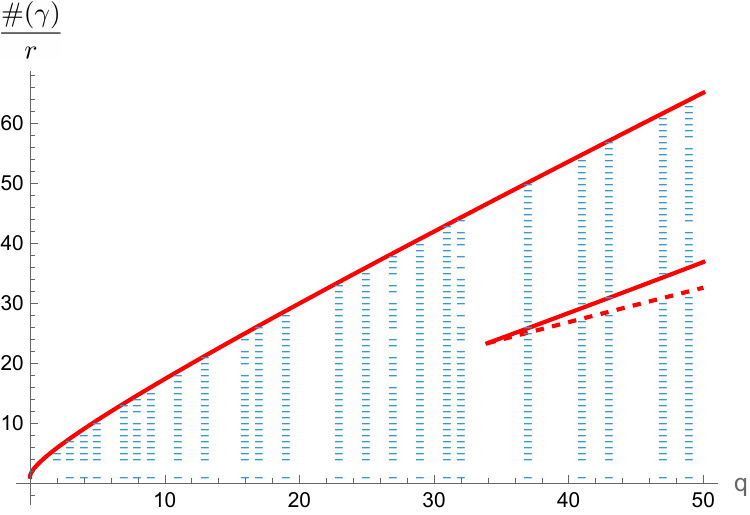}
        \caption{$r=1$}
    \end{subfigure}%
    ~ 
    \begin{subfigure}[t]{0.5\textwidth}
        \centering
        \includegraphics[width=\textwidth]{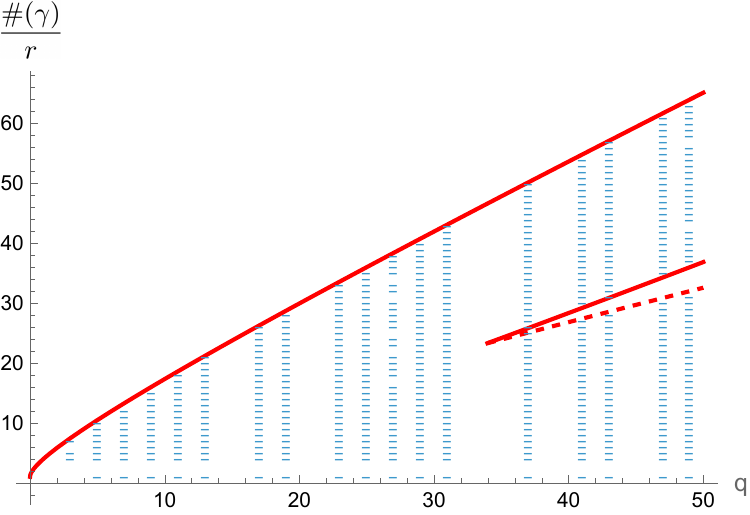}
        \caption{$r=2$}
    \end{subfigure}\\
    \begin{subfigure}[t]{0.5\textwidth}
        \centering
        \includegraphics[width=\textwidth]{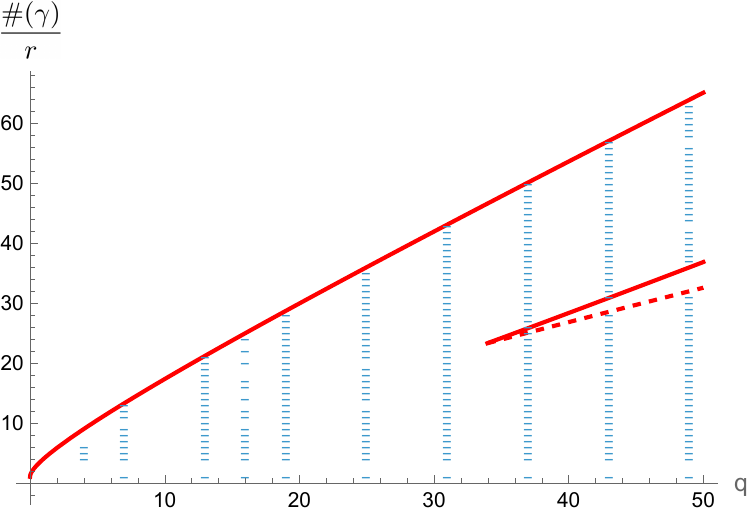}
        \caption{$r=3$}
    \end{subfigure}%
    ~ 
    \begin{subfigure}[t]{0.5\textwidth}
        \centering
        \includegraphics[width=\textwidth]{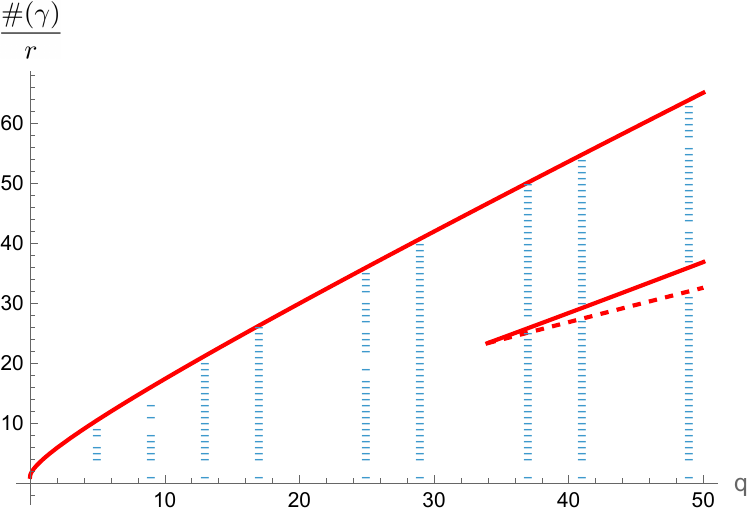}
        \caption{$r=4$}
    \end{subfigure}\\
    \begin{subfigure}[t]{0.5\textwidth}
        \centering
        \includegraphics[width=\textwidth]{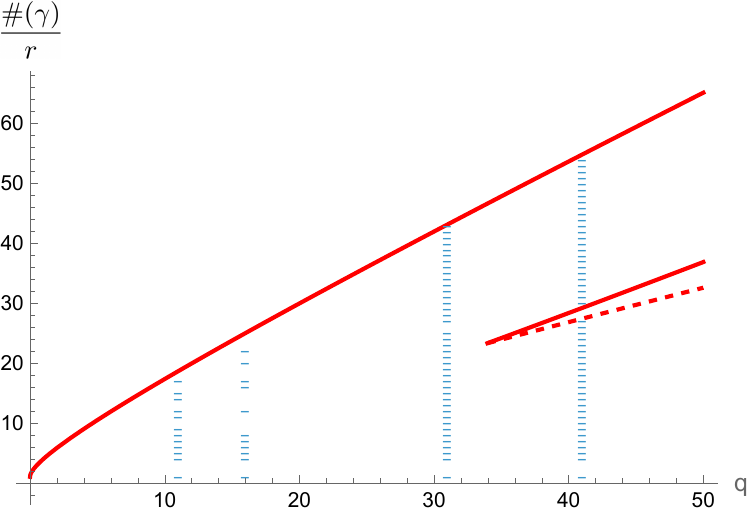}
        \caption{$r=5$}
    \end{subfigure}%
    ~ 
    \begin{subfigure}[t]{0.5\textwidth}
        \centering
        \includegraphics[width=\textwidth]{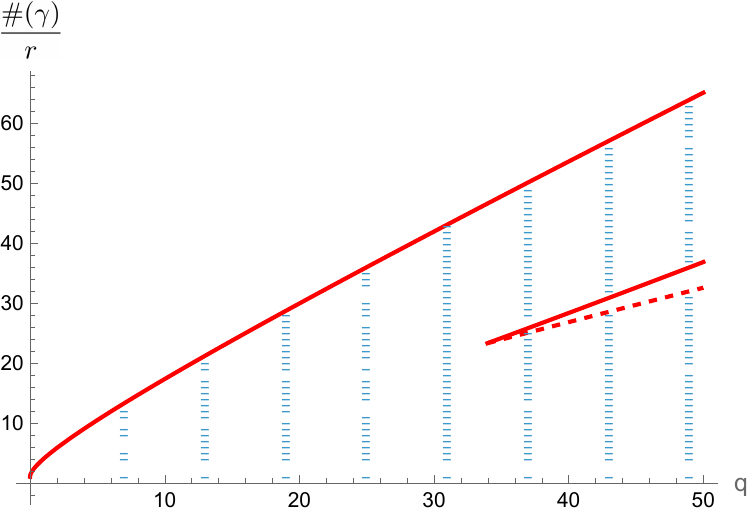}
        \caption{$r=6$}
    \end{subfigure}
    \caption{Figures in \ref{fig:conjI} zoomed in near the origin.}
  \label{fig:conjIzoom}
\end{figure}

\section{Counting points on orbits}
\label{sec:numerical}

In this section, we consider the arithmetic dynamics of Equation \eqref{eq:qp1} in $\mathbb{P}^1\times \mathbb{P}^1$ over a finite field $\mathbb{F}_{q}$. The dynamics becomes ill-defined at eight points, which include
\begin{equation}\label{eq:basepoints}
    \begin{aligned}
    b_1=(\infty,1),\quad b_2=(0,\infty),\quad
    b_3=(0,0),\quad b_4=(\infty,\infty).
\end{aligned}
\end{equation}
The remaining four base points each occur infinitely near one of these. We resolve the dynamics through these points, by using resolution of singularities well known from the construction first provided by Sakai \cite{s:01}. This yields a well-defined notion of orbits and below, we describe algorithms to compute them, which are used to verify Conjecture \ref{conj:numerical}.

\subsection{Resolution of Singularities}\label{sec:resolution}

Equation \eqref{eq:qp1} generally maps points $(x,y)\in \mathbb{F}_q^*\times \mathbb{F}_q^*$ to points $(\overline{x},\overline{y})\in \mathbb{F}_q^*\times \mathbb{F}_q^*$. However, when $s\,x-y=0$, we find that $(\overline{x},\overline{y})=(\infty,1)$. Extending the initial domain to $\mathbb{P}^1\times \mathbb{P}^1$ does not resolve the essential issue that the whole curve defined by $s\,x-y=0$ gets mapped to a single point. 
To resolve this singularity, as well as the others given in Equation \eqref{eq:basepoints},  we introduce four coordinate charts $(x_j,y_j)$, $1\leq j\leq 4$, defined in Equation \eqref{eq:coordcharts} below. 

Each coordinate chart gives rise to a line that was absent in the original space,
\begin{equation*}
     L_j=\{x_j=0,y_j\in \mathbb{A}^1\}   \qquad (1\leq j\leq 4),
\end{equation*}
leading to the \textit{initial value space} of Equation \eqref{eq:qp1} over $\mathbb{F}_q$.
\begin{equation}\label{eq:coordcharts}
\begin{aligned}
    x&=\frac{1}{x_1}, & y&=1+x_1y_1,\\
    x&=x_2(t+x_2y_2), & y&=\frac{1}{x_2},\\
    x&=x_3(t+x_3y_3), & y&=x_3^2(t+x_3y_3),\\
   x&=\frac{1}{x_4(s+x_4y_4)}, & y&=\frac{1}{x_4}.
\end{aligned}
\end{equation}

\begin{definition}\label{defi:initialvaluespace}
Let $X_{t,s}$ denote the variety over $\mathbb{F}_q$ obtained by gluing the lines $L_j$, $1\leq j\leq 4$, to $\{(x,y)\in(\mathbb{A}^1\setminus\{0\})^2\}$ through \eqref{eq:coordcharts}\footnote{Technically, one blows up $\mathbb{P}^1\times \mathbb{P}^1$ at eight base points \eqref{eq:basepoints} and then removes the strict transform of the coordinate lines $x=0$, $x=\infty$, $y=0$, $y=\infty$ \cite{s:01}. This construction may equally well first be done over $\overline{\mathbb{F}}_p$ and then quotiented by the induced action of $\operatorname{Gal}(\overline{\mathbb{F}}_p/\mathbb{F}_q)$, see e.g. \cite[II,\S 4]{mumfordred} for details.}. We call the resulting space the \emph{initial value space} of Equation \eqref{eq:qp1} over $\mathbb{F}_q$.
\end{definition}

On the newly constructed space $X_{t,s}$ all singularities of the evolution are resolved.
To see this, firstly, note that each line $L_j$ replaces a base point $b_j$, for $1\leq j\leq 4$. Returning to the case when $s\,x-y=0$, the image equals $(\overline{x},\overline{y})=b_1$ and, moving to the coordinates that cover $\overline{L}_1$, we have
\begin{equation*}
  \overline{x}_1=0,\quad \overline{y}_1=\frac{t}{x}.
\end{equation*}
This solves the initial issue, since different points on the curve $s\,x-y=0$ now get mapped to different points on $\overline{L}_1$. Next we consider the evolution of points on $L_1$. Taking a point in $L_1=\{x_1=0\}$,  we have
$(\overline{x},\overline{y})=b_2$ and, moving to the coordinates that cover $\overline{L}_2$, we obtain
\begin{equation*}
  \overline{x}_2=0,\quad \overline{y}_2=s\,t(1-s\,y_1).
\end{equation*}

Then, if we take a point in $L_2=\{x_2=0\}$,  we have
$(\overline{x},\overline{y})=b_3$ and, moving to the coordinates that cover $\overline{L}_3$, we get
\begin{equation*}
  \overline{x}_3=0,\quad \overline{y}_3=s\,y_2.
\end{equation*}
For points in $L_3=\{x_3=0\}$, we have
$(\overline{x},\overline{y})=b_4$ and
\begin{equation*}
  \overline{x}_4=0,\quad \overline{y}_4=\frac{s(s\,y_3-t)}{t}.
\end{equation*}
Finally, we consider the evolution of points in $L_4$.
Taking a point in $L_4=\{x_3=0\}$,  we have
\begin{equation*}
    \overline{x}=-\frac{s^2 t}{y_4},\quad
    \overline{y}=1,
\end{equation*}
which lies again in $\mathbb{F}_q^*\times \mathbb{F}_q^*$, unless $y_4=0$ in which case $(\overline{x},\overline{y})$ falls on the line $\overline{L}_1$ and
\begin{equation*}
    \overline{x}_1=0,\quad
    \overline{y}_1=0.
\end{equation*}
We thus see that the evolution is closed on our enlarged space and indeed all singularities have been resolved. In fact, one may show that the evolution \eqref{eq:qp1} lifts uniquely to an isomorphism from $X_{t,s}$ to $X_{st,s}$. The full evolution is schematically displayed in Figure \ref{fig:evolution}.

\begin{figure}[t]
    \centering
\begin{tikzpicture}[>=Stealth, node distance=1.5cm]

  \node (B) {$\mathbb{F}_q^* \times \mathbb{F}_q^*$};
  \node[right=of B] (L1) {$L_1$};
  \node[right=of L1] (L2) {$L_2$};
  \node[right=of L2] (L3) {$L_3$};
  \node[right=of L3] (L4) {$L_4$};

  \draw[->] (B) -- node[above] {\small $s\,x = y$} (L1);
  \draw[->] (L1) -- (L2);
  \draw[->] (L2) -- (L3);
  \draw[->] (L3) -- (L4);

  \draw[->] 
    (L4.north) -- ($ (L4) + (0,1.3) $) -- node[midway, above] {\small $y_4 = 0$}
     ($ (L1) + (0,1.3) $) -- (L1.north)
    ;

  \draw[->] 
    (L4.south) -- ($ (L4) + (0,-1.3) $) -- node[midway, below] {\small $y_4 \neq 0$}
     ($ (B) + (0,-1.3) $) -- (B.south);

  \draw[->] 
    (B.north) -- ($ (B) + (0,1.3) $) -- node[midway, above] {\small $s\,x \neq y  $}
    ($ (B) + (-2.2,1.3) $) -- ($ (B) + (-2.2,0) $)
    -- (B.west);

\end{tikzpicture}
    \caption{Schematic representation of the mapping Equation \eqref{eq:qp1} in $X_{t,s}$.}
    \label{fig:evolution}
\end{figure}
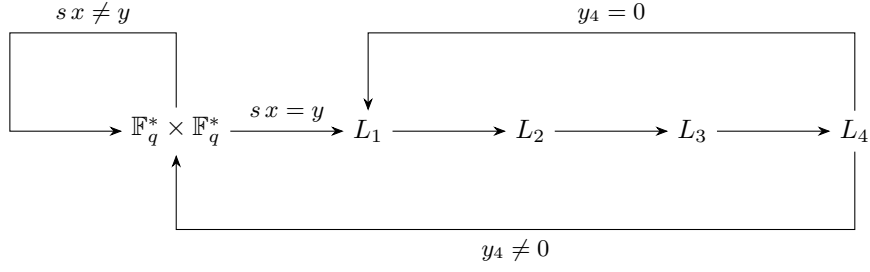

\subsection{Evolution and orbits} We are interested in the evolution of rational points, i.e. points whose coordinates all lie in $\mathbb{F}_q$, in the initial value space, of which there are $(q+1)^2$.
For the purpose of implementing the evolution, we introduce the concept of \textit{states}.
\begin{definition}\label{eq:defi_states}
 We define a \textit{state} $\gamma_*$ as a tuple
\begin{equation*}
    \gamma_*=(j,x,y,t,s),
\end{equation*}
where $0\leq j\leq 4$ an integer, $t,s\in \mathbb{F}_q^*$ and 
\begin{equation}\label{eq:state_condition}
    \begin{cases}
    (x,y)\in \mathbb{F}_q^*\times \mathbb{F}_q^* & \text{if $j=0$,}\\
    (x,y)\in \{0\}\times \mathbb{F}_q & \text{if $1\leq j\leq 4$.}
    \end{cases}
\end{equation}   
\end{definition}
Every state encodes a unique rational point on the initial value space $X_{t,s}$. Here, states with $j=0$ correspond to points in $\mathbb{F}_q^*\times \mathbb{F}_q^*$, whereas states with $j=j_*$, $1\leq j_*\leq 4$, correspond to points in the line
$L_{j_*}$.
The time evolution of a state is then computed by Algorithm \ref{alg:evolution} below.

\begin{algorithm}[H]\caption{Time Evolution}\label{alg:evolution}
\begin{algorithmic}[1]
\Function{Evolution}{$\gamma_*$}\Comment{Input a state.}
     \State{$(j,x,y,t,s)\gets \gamma_*$}
    \If {$j=0$} 
    \If {$s\,x-y\neq 0$}
    \State
    {$\overline{j} \gets 0$ and compute $(\overline{x},\overline{y})$ using equations \eqref{eq:qp1}}
     \Else
     \State 
     {$\overline{j} \gets 1$ and $(\overline{x},\overline{y}) \gets (0,t/x)$}
     \EndIf
     \ElsIf {$j=1$} 
       \State {$\overline{j} \gets 2$, $(\overline{x},\overline{y}) \gets (0,s\,t(1-s\,y))$}
     \ElsIf {$j=2$} 
    \State {$\overline{j} \gets 3$,  $(\overline{x},\overline{y}) \gets (0,s\,y)$}
    \ElsIf {$j=3$} 
    \State {$\overline{j} \gets 4$, $(\overline{x},\overline{y}) \gets (0,s(s\,y-t)/t)$}
    \ElsIf {$j=4$} 
    \If {$y\neq 0$}
    \State 
    {$\overline{j} \gets 0$,  $(\overline{x},\overline{y}) \gets (-s^2 t/y,1)$}
     \Else
     \State 
     {$\overline{j} \gets 1$,  $(\overline{x},\overline{y}) \gets (0,0)$}
     \EndIf
     \EndIf
 \State \textbf{return} $(\overline{j},\overline{x},\overline{y},s\,t,s)$. \Comment{Return evolution of input state}
 \EndFunction
\end{algorithmic}
\end{algorithm}

To compute the orbit of a state, one simply applies the time evolution iteratively until the initial state is reached again. (Note that the reduced length of orbits is on the vertical axis in Figures \ref{fig:conjI} and \ref{fig:conjIzoom}, while it lies on the horizontal axis in Figure \ref{fig:conjII}.)
\begin{algorithm}\caption{Computing an orbit}\label{alg:orbit}
\begin{algorithmic}[1]
\Function{Orbit}{$\gamma_0$}
\Comment{Input a state}
\State{$\gamma \gets (\gamma_0)$}\Comment{Initialise orbit}
\State{$\gamma_*\gets $\Call{Evolution}{$\gamma_0$}}   \Comment{Compute evolution using Algorithm \ref{alg:evolution}}
\While {$\gamma_*\neq \gamma_0 $}
  \State  {\textbf{Append} $\gamma_*$ to $\gamma$}
  \State{$\gamma_* \gets $\Call{Evolution}{$\gamma_*$}}   \Comment{Compute evolution using Algorithm \ref{alg:evolution}}
\EndWhile
 \State{Return $\gamma$}    \Comment{Return orbit}
 \EndFunction
\end{algorithmic}
\end{algorithm}

\subsubsection{Examples} We discuss some examples with $q$ ranging from $2$ to $4$ to illustrate the above definitions. 
Over $\mathbb{F}_2$, the only possible choice for $s$ and $t$ is $t=s=1$, and there are two orbits, one of length five,
\begin{equation*}
   (0,1,1,1,1)\mapsto (1,0,1,1,1)\mapsto (2,0,0,1,1)\mapsto (3,0,0,1,1)\mapsto (4,0,1,1,1),
\end{equation*}
and the other of length four,
\begin{equation*}
   (1,0,0,1,1)\mapsto (2,0,1,1,1)\mapsto (3,0,1,1,1)\mapsto (4,0,0,1,1).
\end{equation*}
 In this case, the upper bound in Conjecture \ref{conj:I} reads
   \begin{equation*}
    \#(\gamma)\leq 3+2\sqrt{2}\approx 5.83,
\end{equation*} 
which is thus sharp.

Over $\mathbb{F}_3$, picking $s=2$ and initial time $t=1$,
 there are a total of three orbits, of lengths $8$, $10$  and $14$, respectively given by
\begin{align*}
&(0, 1, 1, 1, 2)\mapsto (0, 2, 2, 2, 2)\mapsto (0, 2, 2, 1, 2)\mapsto (0, 1, 2, 2, 2)\mapsto 
(1, 0, 2, 1, 2)\mapsto\\
&(2, 0, 0, 2, 2)\mapsto (3, 0, 0, 1, 2)\mapsto (4, 0, 1, 2, 2) );
\end{align*}
\begin{align*}
&(2, 0, 2, 1, 2)\mapsto (3, 0, 1, 2, 2)\mapsto (4, 0, 0, 1, 2)\mapsto (1, 0, 0, 2, 2)\mapsto (2, 0, 1, 1, 2)\mapsto\\
&(3, 0, 2, 2, 2)\mapsto (4, 0, 2, 1, 2)\mapsto (0, 1, 1, 2, 2)\mapsto (0, 1, 2, 1, 2)\mapsto (1, 0, 1, 2, 2) ;
\end{align*}
\begin{align*}
&(3, 0, 2, 1, 2)\mapsto (4, 0, 0, 2, 2)\mapsto (1, 0, 0, 1, 2)\mapsto (2, 0, 2, 2, 2)\mapsto (3, 0, 1, 1, 2)
\mapsto\\
&(4, 0, 2, 2, 2)\mapsto (0, 2, 1, 1, 2)\mapsto (1, 0, 2, 2, 2)\mapsto (2, 0, 0, 1, 2)\mapsto (3, 0, 0, 2, 2)\mapsto\\
&(4, 0, 1, 1, 2)\mapsto (0, 2, 1, 2, 2)\mapsto (1, 0, 1, 1, 2)\mapsto (2, 0, 1, 2, 2).
\end{align*}
The corresponding reduced orbit lengths are $4$, $5$ and $7$ and the upper bound in Conjecture \ref{conj:I} reads
   \begin{equation*}
    \frac{\#(\gamma)}{2}\leq 4+2\sqrt{3}\approx 7.46,
\end{equation*} 
which is thus sharp.

As a last example, we consider the non-prime finite field $\mathbb{F}_4=\{0,1,a,a^2\}$, where $a^2+a+1=0$. The two non-autonomous choices $s=a$ and $s=a^2$ are related by the Frobenius automorphism $z\mapsto z^2$ of $\mathbb{F}_4$. We therefore just consider $s=a$, which has multiplicative order $r=3$. There are five orbits, two of length $12$, one of length $15$ and two of length $18$. We give one of the orbits of length $18$,
\begin{align*}
    & (0, a, a, 1, a) & &\mapsto & &(0, a, a, a, a)& &\mapsto & &(0, a^2, a, a^2, a)& &\mapsto & &(0, a, a^2, 1, a)& &\mapsto &  & \\
    &(1, 0, a^2, a, a)& &\mapsto & &(2, 0, 0, a^2, a)& &\mapsto & &(3, 0, 0, 1, a)& &\mapsto & &(4, 0, a, a, a)& &\mapsto & &\\
    &(0, a^2, 1, a^2, a)& &\mapsto &  &(1, 0, 1, 1, a)& &\mapsto & &(2, 0, 1, a, a)& &\mapsto & &(3, 0, a, a^2, a)& &\mapsto & &\\
    &(4, 0, 0, 1, a)& &\mapsto & &(1, 0, 0, a, a)& &\mapsto & &(2, 0, a^2, a^2, a)& &\mapsto &  &(3, 0, 1, 1, a)& &\mapsto & &\\
    &(4, 0, 1, a, a)& &\mapsto & &(0, 1, 1, a^2, a).  & & & &  & &
\end{align*}
The corresponding reduced orbit lengths are $4$, $5$ and $6$ and
the upper bound in Conjecture \ref{conj:I} reads
   \begin{equation*}
    \frac{\#(\gamma)}{3}\leq 9,
\end{equation*} 
which thus holds but is not sharp.

Upon choosing some values for $s,t\in \mathbb{F}_q^*$, we can compute the orbit starting with any point in $X_{t,s}$ using Algorithm \ref{alg:orbit}. We are interested in getting all the lengths of orbits constructed this way, without double counting orbits that are the same as sets. We do this using Algorithm \ref{alg:obtainorbitlengths} below.

\begin{algorithm}[H]\caption{Obtaining Orbit Lengths}\label{alg:obtainorbitlengths}
\begin{algorithmic}[1]
\Function{Lengths}{$s,t$}
\Comment{Input $s,t\in \mathbb{F}_q^*$}
\State{$R\gets \{(j,x,y,t,s):\text{Eq. \eqref{eq:state_condition} holds}, 0\leq j\leq 4\}$} \Comment{$R$ represents the set of rational points on $X_{t,s}$}
\State{$L \gets ()$}   \Comment{Initialise empty list to store orbit lengths}
\While {$R$ is not empty}
\State{Take a $\gamma_0\in R$}
\State{$\gamma \gets $ \Call{Orbit}{$\gamma_0$}}\Comment{Use Algorithm \ref{alg:orbit} to compute orbit}
\State{\textbf{Append} length $\#(\gamma)$ of orbit to $L$}
  \State{Remove all states in $R$ that occur in orbit $\gamma$.}
\EndWhile
 \State{Return $L$}    \Comment{Return list of orbit lengths}
 \EndFunction
\end{algorithmic}
\end{algorithm}

\subsection{Data on orbit lengths}\label{sec:data}
We implemented Algorithm \ref{alg:obtainorbitlengths} in Magma and used it to obtain all reduced orbit lengths of Equation \eqref{eq:qp1} over $\mathbb{F}_q$, for prime powers $2\leq q\leq 499$. See Section \ref{sec:codedata} for the corresponding data.

We sorted this data by the multiplicative order of $s$. Recall that, necessarily the multiplicative order $r$ of $s$ has to be a divisor of $q-1$. Conversely, for any divisor $r$ of $q-1$, there are $\varphi(r)$ elements $s\in \mathbb{F}_q^*$ with multiplicative order $r$, where $\varphi$ Euler's totient function.

For any prime power $2\leq q\leq 499$, for any divisor $r\,|\,q-1$, we obtained all reduced orbit lengths. For every prime power $2\leq q\leq 499$, we saved the data in a separate .txt file as a nested list,
\begin{equation*}
    ((r,L_r):r\,|\,q-1),
\end{equation*}
where, for every $r\,|\,q-1$, the list $L_r$ is a frequency list of reduced orbit lengths. Namely, its entries are tuples $(l,f_l)$, with $l$ a reduced orbit length and $f_l$ its absolutely frequency.

For example, the data for $q=5$ reads
\begin{verbatim}
 { 
 { 1 , {{ 1, 3 },{ 4, 6 },{ 5, 5 },{ 6, 3 },
        { 7, 3 },{ 8, 2 },{ 9, 3 },{ 10, 1 }}} ,
 { 2 , {{ 1, 1 },{ 4, 3 },{ 5, 2 },{ 6, 1 },
        { 7, 1 },{ 8, 1 },{ 9, 2 },{ 10, 1 }}} , 
 { 4 , {{ 4, 4 },{ 5, 2 },{ 6, 2 },{ 8, 2 },{ 9, 2 }}} 
 }
\end{verbatim}
The divisors of $5-1=4$ are $r=1,2,4$, and for each divisor, the list contains a corresponding entry with  a frequency list of reduced orbit lengths. For instance, for $r=2$ there are $3$ orbits of reduced length $4$. The lists are demarcated by braces so that they can easily be loaded in Mathematica for analysis. See Section \ref{sec:codedata} for the corresponding Mathematica code.

\subsection{An orbit length distribution}
From the orbit length data, we can compute the distribution of orbit lengths with respect to choosing an initial point in the initial value space $X_{t,s}$ uniformly random. We consider the example $q=p=499$, with $t=1$  and $s=140$, so that $s$ has multiplicative order $r=6$.

Conjecture \ref{conj:II} then says that all reduced orbit lengths lie in the $M_q=6$ bins,
\begin{equation}\label{eq:bins499}
\begin{aligned}
    B_1&=[456,544], &
    B_2&=[228,272], &
    B_3&=[152,181], \\
    B_4&=[114,136], &
    B_5&=[92,108], &
    B_6&=[1,90],
\end{aligned}
\end{equation}
which is indeed the case.

In Figure \ref{fig:conjII}, the absolute frequencies of all reduced orbits lengths are displayed. There is only one orbit with reduced length $544$. The corresponding absolute frequency in Figure \ref{fig:conjII} is thus $544$ since there are $544$ choices of initial values for which the corresponding orbit has this reduced length. Similarly, there are two orbits of reduced length $543$ and the corresponding absolute frequency is thus $1086$. 

The absolute frequencies of reduced orbit lengths occurring in bins $B_1$,\ldots,$B_6$ are respectively
\begin{equation*}
    113831,\quad 54446,\quad 13551,\quad 19140,\quad 6935,\quad 42097,
\end{equation*}
totalling $250000$.
This means for example that, taking initial values uniformly random, the chance of the  reduced length of the corresponding orbit falling in $B_1$ is $113831/250000\approx 46\%$.



\begin{figure}[t!]
    \centering
    \begin{subfigure}[t]{0.48\textwidth}
        \centering
        \includegraphics[width=\textwidth]{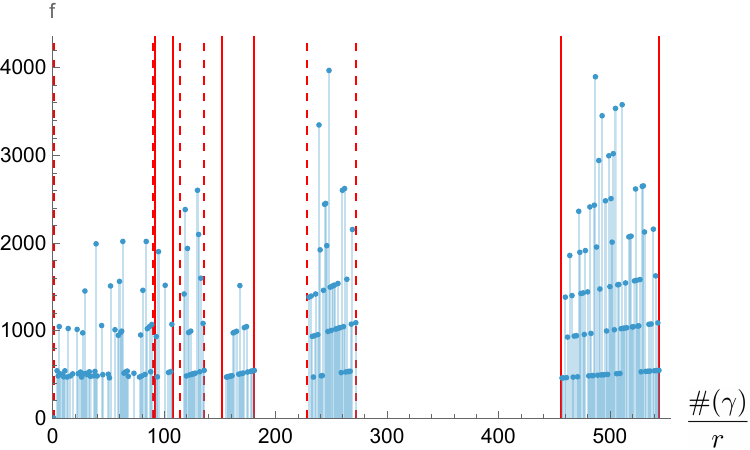}
    \end{subfigure}%
    ~ 
    \begin{subfigure}[t]{0.48\textwidth}
        \centering
        \includegraphics[width=\textwidth]{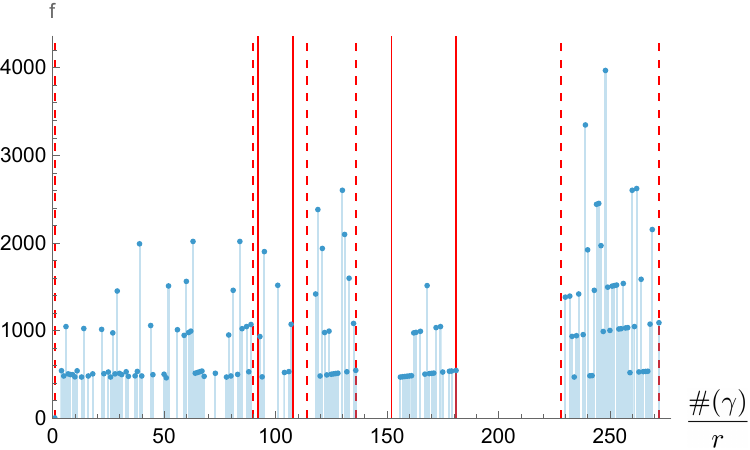}
    \end{subfigure}\vspace{5mm}\\
    \begin{subfigure}[t]{0.48\textwidth}
        \centering
        \includegraphics[width=\textwidth]{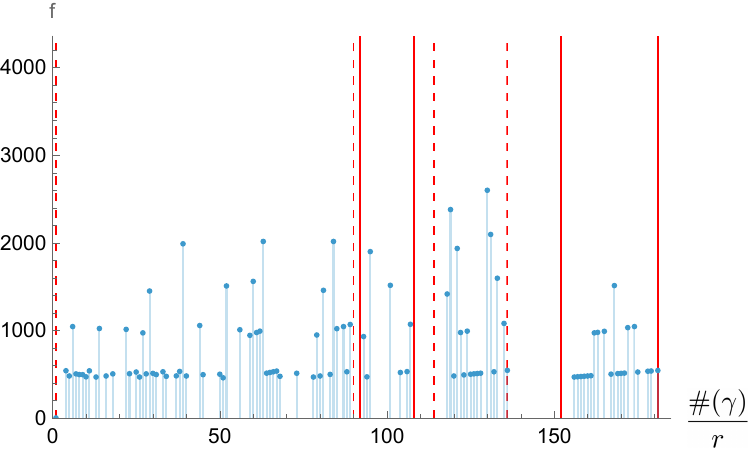}
    \end{subfigure}%
    ~ 
    \begin{subfigure}[t]{0.48\textwidth}
        \centering
        \includegraphics[width=\textwidth]{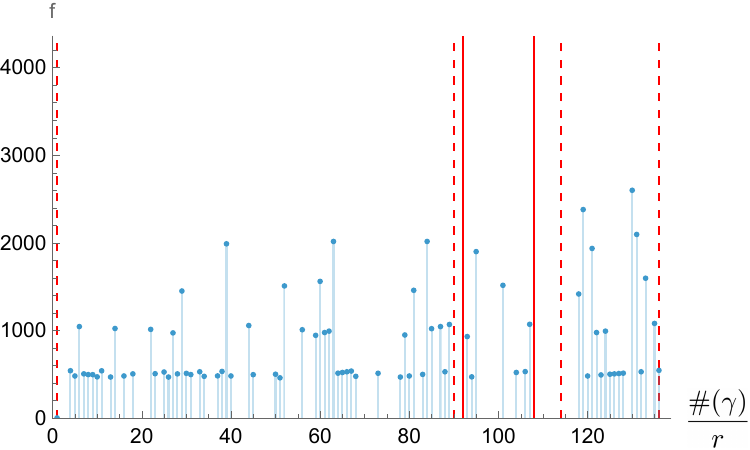}
    \end{subfigure}\vspace{5mm}\\
    \begin{subfigure}[t]{0.48\textwidth}
        \centering
        \includegraphics[width=\textwidth]{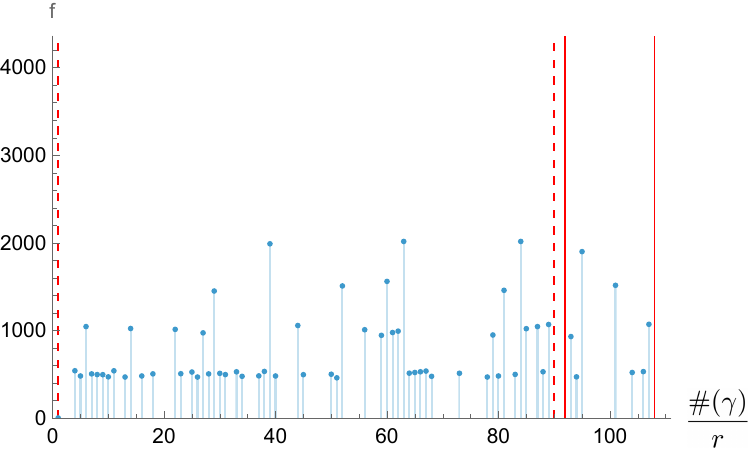}
    \end{subfigure}%
    ~ 
    \begin{subfigure}[t]{0.48\textwidth}
        \centering
        \includegraphics[width=\textwidth]{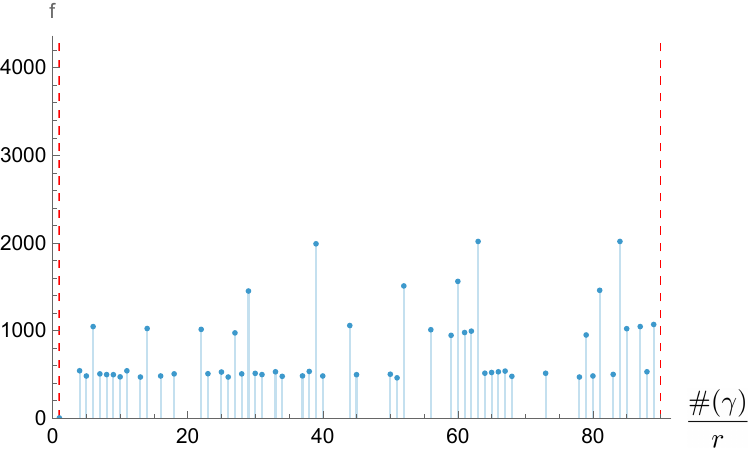}
    \end{subfigure}
    \caption{In these plots, for $q=p=499$, $t=1$ and $s=140$, which has multiplicative order $r=6$ in $\mathbb{F}_q$, in blue the absolute frequency $f$ of orbit lengths over $r$ for all choices of initial values as well as the six bins  \eqref{eq:bins499} demarcated alternately by solid and dashed red grid lines. In particular, taking random initial values uniformly, the plots show the unnormalised distribution of the length of the corresponding orbit over $r$.
    In the top-left plot the full distribution is shown, whereas the other plots show the distribution closer to the origin at different scales.
    }
  \label{fig:conjII}
\end{figure}

\section{Integrals of motion}\label{sec:integral}
When considering Equation \eqref{eq:qp1} over a finite field, $s$ necessarily has finite multiplicative order, say $r$. This means in particular that $s$ is a root of the $r$th cyclotomic polynomial, $\Phi_r(s)$.
Our main theoretical result, shows that Equation \eqref{eq:qp1} has an integral of motion,  whenever $s$ has finite multiplicative order, regardless of the underlying field.

\subsection{Construction of integrals of motion}

\begin{theorem}\label{thm:integrals}
Let $r\geq 1$ be any positive integer,  $R=\mathbb{Z}[s]/(\Phi_r(s))$ denote the $r$th cyclotomic integers and $Q$ its field of fractions. Then the Laurent polynomial $I_r\in R[x^{\pm 1},y^{\pm 1},t]$, defined by the trace
\begin{equation}\label{eq:defiintegrals}
    I_r=\operatorname{Tr}[A(s^{r-1})\cdot A(s^{r-2})\cdot\ldots\cdot A(s^2)\cdot A(s)\cdot A(1)]-(t^{r}+1),
\end{equation}
where
\begin{equation}\label{eq:Amatrix}
 A(z)=A_0+z\,A_1+z^2 A_2,\qquad A_2=\begin{bmatrix}
 1 & 0\\
 0 & 0
 \end{bmatrix},
\end{equation}
with
\begin{equation*}
\begin{aligned}
A_0&=\begin{bmatrix}
    t+x-xy & -w\,x\\
    w^{-1}(t+x-ty-2xy+xy^2) & x(y-1)
\end{bmatrix},\\   
A_1&=\begin{bmatrix}
    y-x+x/y-1-t/x & w\\
    w^{-1}(y-2x-1+xy+x/y-t/x) & 1
\end{bmatrix},
\end{aligned}
\end{equation*}
defines an integral of motion of $q\Pone$, i.e. $\overline{I}_r=I_r$ in $Q(x,y,t)$,
that is of bidegree $(2r,2r)$ in $(x,y)$.
\end{theorem}
\begin{remark}
It is important to emphasise that $I_r\in R[x^{\pm 1},y^{\pm 1},t]$, rather than just $I_r\in Q[x^{\pm 1},y^{\pm 1},t]$, since we are interested in reductions of these integrals modulo primes. Such reductions are not always well-defined for elements of the latter ring. Furthermore,
\begin{equation}\label{eq:Ileading}
    I_r=-x^r\left(1-\frac{1}{y}\right)^r-\frac{t^r}{x^r}+(-1)^{r+1}y^r+J,
\end{equation}
where the remainder $J$ is a polynomial in $x,x^{-1},y,y^{-1}$ of degree less or equal to $r-1$ in each separate variable, so that there is no loss of degree under reductions modulo primes.
\end{remark}
\begin{remark}
        The matrix $A(z)$ is the coefficient matrix of the spectral part of the Lax pair for $q\Pone$ taken from \cite{murata}, normalised appropriately for our purposes. Its dependence on $w$ drops out when taking the trace, so that the integrals of motion are independent of $w$.
\end{remark}
\begin{remark}
    In the proof of Theorem \ref{thm:integrals}, we show that
    \begin{equation*}
\operatorname{Tr}[A(s^{r-1}z)\cdot A(s^{r-2}z)\cdot\ldots\cdot A(s^2z)\cdot A(sz)\cdot A(z)]=t^r+I_r z^r+z^{2r},        
    \end{equation*}
is an invariant of motion, from which it follows that $I_r$ is. It further yields the trivial invariant $t^r$. The determinant of this product of $A$ matrices is also trivially invariant, as follows from $|A(z)|=z^3$.
\end{remark}
\begin{remark}
We note that certain special cases of integrals of motion for $q$-Painlev\'e equations have appeared in the literature before \cite{grammaticosetal2011,carsteatakenawa}; see also \cite{wei} for an example involving the elliptic Painlev\'e equation. In \cite{grammaticosetal2011}, the cases when $s$ is a second or third root of unity over $\mathbb{C}$ were given.
In \cite{carsteatakenawa}, it is asserted that 
 the $r$-fold composition of a $q$-Painlev\'e equation preserves a Halphen pencil when $s$ is an $r$th root of unity in $\mathbb{C}$. This yields a different construction of integrals of motion, by solving the linear system $|-r K_X|$, where $-K_X$ is the anti-canonical divisor of the corresponding initial value space.
 Neither of these studies considered dynamics over a finite field. 
\end{remark}

Here are the first few integrals of motion,
\begin{align*}
I_1=&-x+y+\frac{x}{y}-\frac{t}{x},\\
I_2=&-\frac{t^2}{x^2}+\frac{2 t y}{x}+\frac{2 t}{y}-\frac{x^2}{y^2}+\frac{2 x^2}{y}-x^2-2 x y+2 x-y^2,\\
I_3=&-3 s x^2 y+3 s t y-\frac{3 s x^2}{y}+6 s x^2-3 s x y^2+3 s x y-\frac{t^3}{x^3}+\frac{3 t^2 y}{x^2}+\frac{3 t^2}{x y}-\frac{3 t y^2}{x}\\
&-\frac{3 t x}{y^2}+\frac{3 t x}{y}+\frac{x^3}{y^3}-\frac{3 x^3}{y^2}+\frac{3 x^3}{y}-x^3-3 x^2 y-\frac{3 x^2}{y}+6 x^2+y^3,\\
I_4=&-\frac{4 s t x}{y}+\frac{4 s t^2 y}{x}+4 s t x-8 s t y^2+4 s t y+\frac{4 s x^3}{y^2}-4 s x^3 y-\frac{12 s x^3}{y}+12 s x^3+4 s x y^3\\
&-4 s x y^2-\frac{t^4}{x^4}+\frac{4 t^3 y}{x^3}+\frac{4 t^3}{x^2 y}-\frac{6 t^2 y^2}{x^2}-\frac{4 t^2}{x}-\frac{6 t^2}{y^2}+\frac{4 t^2}{y}+\frac{4 t x^2}{y^3}-\frac{8 t x^2}{y^2}+\frac{4 t x^2}{y}\\
&+\frac{4 t y^3}{x}-4 t x y+4 t x-\frac{x^4}{y^4}+\frac{4 x^4}{y^3}-\frac{6 x^4}{y^2}+\frac{4 x^4}{y}-x^4+6 x^2 y^2-12 x^2 y+6 x^2-y^4,
\end{align*}
where, respectively, $s=1$, $s=-1$, $s^2+s+1=0$, $s^2+1=0$.

The construction of the integrals of motion in Theorem \ref{thm:integrals} relies on a Lax pair of $q\Pone$ \cite{murata}. In preparation of the proof of Theorem \ref{thm:integrals} we recall it here. It is given by
\begin{equation}\label{eq:laxpair}
\begin{aligned}
    Y(qz)&=A(z)Y(z),\\
    \overline{Y}(z)&=B(z)Y(z),
\end{aligned}
\end{equation}
where the matrix polynomial $A(z)$ is defined in \eqref{eq:Amatrix} and $B(z)$ is given by
\begin{equation*}
    B(z)=I+z^{-1}B_0,\qquad B_0=s\begin{bmatrix}
        x(1-y^{-1}) & w\,x \,y^{-1}\\
        -w^{-1} x\,y (1-y^{-1})^2 & -x(1-y^{-1})
    \end{bmatrix}.
\end{equation*}
Compatibility of \eqref{eq:laxpair} is equivalent to
\begin{equation*}
    \overline{A}(z)B(z)-B(sz)A(z)=0,
\end{equation*}
which in turn is equivalent to the evolution of $(x,y,t)$ defined by Equation \eqref{eq:qp1}  and $w\mapsto \overline{w}$,
\begin{equation}\label{eq:gaugeevolution}
    \overline{w}=\left(1- \frac{s\,x}{y}\right)w.
\end{equation}
Put differently, the evolution of $(x,y,t)$ and $w$ defined by $q\Pone$ and the above equation yields the evolution of the coefficient matrix $A(z)\mapsto \overline{A}(z)$, where
\begin{equation*}
    \overline{A}(z)=B(sz)A(z)B(z)^{-1}.
\end{equation*}
We are now ready to prove Theorem \ref{thm:integrals}.
\begin{proof}[Proof of Theorem \ref{thm:integrals}]
   Take any $r\geq 1$, denote by $R=\mathbb{Z}[s]/(\Phi_r(s))$ the $r$th cyclotomic integers and  by $Q=\mathbb{Q}[s]/(\Phi_r(s))$ the $r$th cyclotomic field. We treat $x,y,t,w$ as formal variables and in particular consider $A(z)$ and $B(z)$ as Laurent polynomials in $z$ with coefficients from the ring of $2\times 2$ matrices with entries from the field $Q(x,y,t,w)$.

    Construct the matrix
    \begin{equation}\label{eq:defi_M}
        M(z)=A(s^{r-1}z)\cdot A(s^{r-2}z)\cdot\ldots\cdot A(sz)\cdot A(z).
    \end{equation}
This is a matrix polynomial of degree $2\,r$ in $z$ that satisfies
  \begin{equation}\label{eq:Mevolution}
      \begin{aligned}
          M(sz)&=A(z) M(z) A(z)^{-1},\\
          \overline{M}(z)&=B(z) M(z) B(z)^{-1},
      \end{aligned}
  \end{equation}
  where  $\overline{\cdot}$ acts on the coefficients through the automorphism of $Q(x,y,t,w)$, that leaves $Q$ invariant, and acts on the variables by $(x,y,t,w)\mapsto (\overline{x},\overline{y},\overline{t},\overline{w})$ with images defined by equations \eqref{eq:qp1} and \eqref{eq:gaugeevolution}.

We now turn to its trace,
\begin{equation*}
    \mathcal{I}_r(z)=\operatorname{Tr} M(z),
\end{equation*}
which is a degree $2r$ polynomial in $z$ with coefficients in $Q(x,y,t)$. Due to equations \eqref{eq:Mevolution}, it satisfies
      \begin{align}
          \mathcal{I}_r(sz)&=\mathcal{I}_r(z),\label{eq:Ievolution1}\\
          \overline{\mathcal{I}}_r(z)&=\mathcal{I}_r(z).\label{eq:Ievolution2}
      \end{align}
By equation \eqref{eq:Ievolution1}, $\mathcal{I}_r(z)$ must a polynomial in $z^r$, so
\begin{equation*}
    \mathcal{I}_r(z)=b_0+b_1 z^r+b_2 z^{2r},
\end{equation*}
for some $b_0,b_1,b_2\in Q(x,y,t)$ that, by equation \eqref{eq:Ievolution2}, are each integrals of motion of $q\Pone$.  We will show that
\begin{equation}\label{eq:integralz}
    \mathcal{I}_r(z)=t^r+I_r z^r+z^{2r}.
\end{equation}
Firstly, since the eigenvalues of $A_2$ are $\{0,1\}$, we have
\begin{equation*}
    b_2=\operatorname{Tr} A_2^r=1^r+0^r=1.
\end{equation*}
Secondly, since the eigenvalues of $A_0$ are $\{0,t\}$,
\begin{equation*}
    b_0=\operatorname{Tr} A_0^r=t^r.
\end{equation*}
Finally, by the definition of $I_r$, see equation \eqref{eq:defiintegrals},
\begin{equation}\label{eq:directformulaintegral}
    I_r=\mathcal{I}_r(1)-(t^r+1)=b_1.
\end{equation}
This gives \eqref{eq:integralz} and shows that $I_r$
is an integral of motion of $q\Pone$.

Note that $I_r$ is a Laurent polynomial in $x$, $y$, and polynomial in $t$, with coefficients from $R$, since the same is true for $A(z)$, ignoring the dependence of $A(z)$ on $w$. In other words, $I_r\in R[x^{\pm 1},y^{\pm 1},t]$. It remains to be shown that $I_r$ has bidegree $(2r,2r)$ in $(x,y)$. We first consider the degree in $x^{-1}$. For convenience, we set $w=1$ so that
\begin{equation*}
    A(z)=-\frac{t\,z}{x}\begin{bmatrix}
        1 & 0\\
        1 & 0
    \end{bmatrix}+\mathcal{O}(1),
\end{equation*}
as $x\rightarrow 0$. Therefore
\begin{equation*}
    M(z)=(-1)^r s^{\frac{1}{2}r(r-1)}\frac{t^r\,z^r}{x^r}\begin{bmatrix}
        1 & 0\\
        1 & 0
    \end{bmatrix}+\mathcal{O}(x^{-(r-1)}),
\end{equation*}
as $x\rightarrow 0$. Noting that $(-1)^r s^{\frac{1}{2}r(r-1)}=-1$, we thus obtain
\begin{equation*}
    \mathcal{I}_r(z)=-\frac{t^r\,z^r}{x^r}+\mathcal{O}(x^{-(r-1)}),
\end{equation*}
and, by equation \eqref{eq:directformulaintegral}, 
\begin{equation*}
    I_r=-\frac{t^r}{x^r}+\mathcal{O}(x^{-(r-1)}),
\end{equation*}
as $x\rightarrow 0$.

To work out the degree of $I_r$ in $x$, it is convenient to set $w=y-1$, so that
\begin{equation*}
    A(z)=x\left(1-\frac{1}{y}\right)\begin{bmatrix}
        -(z+y) & -y\\
        +(z+y) & +y
    \end{bmatrix}+\mathcal{O}(1),
\end{equation*}
as $x\rightarrow \infty$. A similar calculation as above then gives
\begin{equation*}
    I_r=-\left(1-\frac{1}{y}\right)^r x^r+\mathcal{O}(x^{r-1}),
\end{equation*}
as $x\rightarrow \infty$.

To estimate the degree in $y^{-1}$, it is convenient to set $w=1$ and one quickly obtains
\begin{equation*}
    I_r=(-1)^{r+1} \frac{x^r}{y^r}+\mathcal{O}(y^{-(r-1)}),
\end{equation*}
as $y\rightarrow 0$.
  
Finally, to estimate the degree in $y$, it is convenient to set $w=y$, so that
\begin{equation*}
    A(z)=y\begin{bmatrix}
        z-x & z-x\\
        x & x
    \end{bmatrix}+\mathcal{O}(1),
\end{equation*}
as $y\rightarrow \infty$, and from this one obtains
\begin{equation*}
    I_r=(-1)^{r+1} y^r+\mathcal{O}(y^{r-1}).
\end{equation*}
In summary, we have shown that $I_r$, as a polynomial in $x,x^{-1},y,y^{-1}$, has degree $r$ with respect to each of the four variables individually. Even stronger, we have derived equation \eqref{eq:Ileading}. In particular $I_r$ has bidegree $(2r,2r)$ in $(x,y)$ and the theorem follows.
\end{proof}

\subsection{Genera of fibres}\label{sec:genus}

 As a corollary of Theorem \ref{thm:integrals}, any reduced orbit of Equation \eqref{eq:qp1} over a finite field lies in one of the fibres of the integral $I_r$,
   \begin{equation}\label{eq:fibres}
    \{(x,y)\in X_{t,s}:I_r(x,y)=c\},\qquad (c\in \mathbb{A}^1).
   \end{equation}
   To our surprise, the geometric genus of the fibres is generically one, so that 
   these higher degree equations essentially define elliptic curves.
\begin{conjecture}\label{conj:genus}
For fixed values of $s$ and $t$, the geometric genus of the fibres \eqref{eq:fibres} is generically one and reduces to zero only on the closed subset of $\{c\in\mathbb{A}^1\}$ defined by
\begin{equation}\label{eq:genuszerocondition}
    c^4+(-1)^r c^3 - 8\, t^r c^2 - (-1)^r 36\, t^r c+ 16\, t^{2r} - 27\, t^r=0.
\end{equation}
\end{conjecture}
\begin{remark}
   This conjecture does not contradict the genus-degree formula, since the curves defined by the polynomials equations in $\mathbb{P}^1\times \mathbb{P}^1$
   are singular at the base points \eqref{eq:basepoints}.
\end{remark}
\begin{remark}
    We note that Conjecture \ref{conj:genus} implies the Hasse upper bound in Conjecture \ref{conj:numerical}.
\end{remark}

For every $r\geq 1$, the $r$th power of $q\Pone$, with $\Phi_r(s)=1$, is an autonomous mapping that leaves invariant the pencil defined by the following polynomial equation,
\begin{equation}\label{eq:defiM}
    x^r y^r I_r(x,y)-c\,x^r y^r=0.
\end{equation}
Consider the curve $C$ this polynomial defines in $\{(x,y)\in\mathbb{P}_{\mathbf{k}}^1\times \mathbb{P}_{\mathbf{k}}^1\}$ over the field
 $\mathbf{k}=Q(t,c)$ of rational functions in $t,c$ over $Q=\mathbb{Q}[s]/(\Phi_r(s))$. We computed its geometric genus in Magma for $1\leq r\leq 4$ and confirmed that it is one in each case. This implies that the geometric genera of the fibres \eqref{eq:fibres} are one for generic $c\in\mathbb{A}^1$ and $t\in\mathbb{A}^1\setminus \{0\}$ over finite fields.
Conjecture \ref{conj:genus} was checked in Magma for $1\leq r\leq 6$, with integer values for $c$ and $t$ ranging from $0$ to $10$ and $1$ to $10$ respectively in each case. See Section \ref{sec:codedata} for the corresponding code.

The condition in equation \eqref{eq:genuszerocondition} comes from considering the spectral equation of the matrix polynomial $M(z)$ defined in equation \eqref{eq:defi_M},
\begin{equation*}
0=|\lambda I-M(z)|=\lambda^2-\lambda\operatorname{Tr}M(z)+|M(z)|.
\end{equation*}
Using equation \eqref{eq:integralz} and $|A(z)|=z^3$, which implies $|M(z)=(-1)^{r+1}z^{3r}$,  the spectral equation can be written explicitly as
\begin{equation*}
    \lambda^2-(t^r+c\,z^r+z^{2r})\lambda+(-1)^{r+1}z^{3r}=0.
\end{equation*}
A direct calculation shows that the curve defined by the spectral equation in $(z,\lambda)$, is singular if and only if equation \eqref{eq:genuszerocondition} holds.  Given a singularity $(z_0,\lambda_0)$, there is a matrix $M(z)$ satisfying $M(z_0)=\lambda_0 I$, for some corresponding value of $(x,y)\in X_{t,s}$ which forms an additional singularity of the fibre and consequently reduces the geometric genus to zero. 

Orbits made out of such singularities of fibres constitute algebraic solutions of $q\Pone$.
For example, in the autonomous case $s=1$, solving the equation
\begin{equation*}
    A(z_0)=\lambda_0 I,
\end{equation*}
gives $z_0=x$, $\lambda_0=y^3$ and
\begin{equation*}
    x=y^2,\quad y^4-y^3-t=0.
\end{equation*}
This defines an algebraic solution of $q\Pone$, that was also derived in \cite{ohyama2010}\footnote{To translate between \eqref{eq:qp1} and the version of $q\Pone$ in \cite{ohyama2010}, we note that eliminating $x$ from \eqref{eq:qp1} yields  \begin{equation*}
 \overline{y}\,\underline{y}=\frac{s\,t}{y(y-1)},
 \end{equation*}
 which transforms to $\overline{f}f^2\underline{f}=1/t(1-f)$ under $y=1/\underline{f}$. The further change $t\mapsto 1/t$ gives the version of $q\Pone$ in \cite{ohyama2010}.}.

Similarly, when $s=-1$, the equation
\begin{equation*}
    M(z_0)=\lambda_0 I,\qquad M(z)=A(-z)A(z),
\end{equation*}
leads to the algebraic solution
\begin{equation*}
    x=\frac{y^2}{1-2y},\quad y^4-y^3+4\,t\,y^2-4\,t\,y+t=0,
\end{equation*}
that also appeared in \cite{ohyama2010}.







\section{Conclusion}\label{s:conc}
In this paper, we studied the lengths of orbits of the discrete Painlev\'e equation \eqref{eq:qp1} in $\mathbb F_q$, for prime powers $q$ ranging from 2 to 499. Based on the evidence provided by over 200M orbits in our study, we conjecture that the orbits of this dynamical system must lie on algebraic curves of genus 1. By using the coefficient matrix \eqref{eq:Amatrix}, in its associated Lax pair, we were able to construct the polynomials defining these algebraic curves explicitly.
Although extensions to other discrete Painlev\'e equations were outside the scope of this paper, we suggest that the latter construction can be equally well be able to be carried out for such equations with discrete Lax pairs.

Intriguing observations of the space of matrices $M(z)$ (see equation \eqref{eq:defi_M}), fibred by their traces, suggest possible pathways to proving our Conjectures \ref{conj:numerical} and \ref{conj:genus}.  Another possible proof could be found through the geometry of the initial value space in the spirit of \cite{carsteatakenawa}. 


There are tantalizing connections to  local Artin-Mazur zeta functions $\zeta_p(\cdot)$, with coefficients determined by orbit lengths. It would be interesting to investigate what role these zeta functions may play in the theory of discrete Painlev\'e equations.

\end{document}